\documentclass[12pt]{amsart}
\usepackage{SummerPreamble}

\begin{document}
 
 \title[Financial value of weak anticipation]{The financial value of knowing the distribution of stock prices in  discrete market models}
  \author[Amiran]{Ayelet Amiran{$^{\dag}$}}
    \address{$^{\dag}$University of Massachusetts, Amherst, MA 01003, U.S.A. }
    \email{ aamiran@umass.edu} 
\author[Baudoin]{Fabrice Baudoin{$^{\dag}$}}
    \address{$^{\dag}$University of Connecticut, Storrs, CT 06269, U.S.A. }
    \email{ fabrice.baudoin@uconn.edu} 
\author[Brock]{Skylyn Brock{$^{\dag}$}}
    \address{$^{\dag}$Colorado Mesa University, Grand Junction, CO 81501, U.S.A. }
    \email{ snbrock@mavs.coloradomesa.edu} 
\author[Coster]{Berend Coster{$^{\dag}$}}
    \address{$^{\dag}$University of Connecticut, Storrs, CT 06269, U.S.A. }
    \email{ berend.coster@uconn.edu} 
\author[Craver]{Ryan Craver{$^{\dag}$}}
    \address{$^{\dag}$University of Maryland, College Park, MD 20742, U.S.A. }
    \email{ rcraver@terpmail.umd.edu} 
\author[Ezeaka]{Ugonna Ezeaka{$^{\dag}$}}
    \address{$^{\dag}$University of Massachusetts, Amherst, MA 01003, U.S.A. }
    \email{ uezeaka@umass.edu} 
\author[Mariano]{Phanuel Mariano{$^{\dag}$}}
    \address{$^{\dag}$University of Connecticut, Storrs, CT 06269, U.S.A. }
    \email{ phanuel.mariano@uconn.edu}
\author[Wishart]{Mary Wishart{$^{\dag}$}}
    \address{$^{\dag}$Eastern Connecticut State University, Willimantic,CT 06226,U.S.A. }
    \email{ wishartmar@my.easternct.edu} 
\date{\today}{\vspace{-5ex}}
\maketitle 

\begin{abstract}
An explicit formula is derived for the value of weak information in a discrete time model that works for a wide range of utility functions including the logarithmic and power utility. We assume a complete market with a finite number of assets and a finite number of possible outcomes. Explicit calculations are performed for a binomial model with two assets. The case of trinomial models is also discussed.
\end{abstract} 
\tableofcontents

\setcounter{page}{1}


\section{Introduction}

Suppose an investor knows the distribution of the prices of the stocks in the market at a future time and this investor wants to optimize  her or his expected utility from wealth at that future time. Our basic question is: \textit{What is the financial value of this information ?}

Much of the research into utility optimization and the financial value of weak information has been looked at previously in a continuous time setting (\cite{modantic} and \cite{weakinfo}). 
The purpose of this paper is to investigate how to optimize a stock portfolio given weak information in a discrete time setting. 

Though trinomial models will be discussed at the end of the paper, in most of this work we will assume that the market is complete. We will also assume that there are no transactions costs. For a definition of complete markets see \cite{book}. The main tool we use in finding the optimal expected utility given the weak information on future stock prices is the martingale method (see \cite{Shreve}).

As with classical results in this field, we will be looking at the expected utility as opposed to the expected wealth. This is an important difference to note since utility functions allow us to include an individual's attitude towards risk. 

\

\textbf{Acknowledgments:} This research was funded by the  NSF grant DMS 1659643. The authors would like to thank Oleksii Mostovyi for several instructive discussions and comments on this work.

\section{Utility Functions}
There are many different utility functions used in mathematics and economics to measure an individuals happiness or satisfaction. We denote our utility functions by $U$. We require that a utility function is strictly concave, strictly increasing, and continuously differentiable. We assume as \cite{modantic} in that \begin{equation}\label{derivlim} \lim\limits_{x\to0}U'(x)=+\infty \text{ and } \lim\limits_{x\to\infty}U'(x)=0. \end{equation} 
These conditions are sufficient for a utility function to exhibit risk aversion, satisfy the law of diminishing marginal utility, and to guarantee that an increase in wealth results in an increase in utility. Further, when discussing the risk aversion of our utility functions, we use the absolute and relative risk aversion functions (see \cite{rra}). We will be looking specifically at three different types of utility functions.
\\\\The three types of utility functions we use are:

\

\begin{enumerate}

    \item\underline{Log Utility}: $U(x) = \ln(x), \quad x>0$  \\\\
    The log utility function has a constant relative risk aversion of 1. This implies the individual will always take on a constant proportion of risk with respect to their wealth. \\
    \item \underline{Power Utility}: $U(x) = \frac{x^\gamma}{\gamma},$\quad for $-\infty<\gamma<0$ and $0<\gamma<1$ and $x>0$\\\\
    The power utility function also has a constant relative risk aversion, but the constant value is $1-\gamma$. Thus, the power utility function is less risk averse compared to log utility function for $0<\gamma<1$. In this case, the constant $\gamma$ reflects the relative risk adversity with the individual becoming more risk adverse as $\gamma$ approaches $0$. If $-\infty<\gamma<0$ the individual is more risk-averse than an individual whose preferences can be described by the logarithmic utility function. As $\gamma$ approaches $-\infty$ the individual becomes more and more risk-averse.\\
    \item \underline{Exponential Utility}: $U(x) = - e^{-\a x},$ for $\a>0$ and $x\in\R$\\\\
    The exponential utility function has a constant absolute risk aversion of 1. Thus, the individual with an exponential utility function will assume a constant amount of risk rather than a constant proportion of risk with respect to their wealth. Notice that the exponential utility function does not satisfy the condition \eqref{derivlim}, but it is still an interesting function to note, and our results still hold true for this function.
\end{enumerate}


\section{Modelling Weak Anticipations on Discrete Time Complete Markets with a Discrete State Space}\label{2}

\subsection{Setup}

Suppose we have a market with $d$ financial assets, and a sample space 
$\Omega=\{\o_1,...,\o_M\}$ of possible outcomes of all the asset prices after one time period. For all probability measures $\P$, we always assume $\P(\o_j)>0, \forall \o_j\in\Omega$. This is not a restriction since if $\P(\o_j)=0$, then we exclude $\o$ from $\Omega$. Then let $N$ be our final time period. Let $\vec{S_n}\in\mathbb{R}^d$ denote the asset prices at time $n$ where $n\in\{0,1,...,N\}$. Further, let the random variable $V_n$ denote the value of the portfolio at time $n$. Denote the initial wealth of the investor $V_0$ by $v$. Without loss of generality we can assume one of the assets is a risk-free asset. We define $r$ to be the rate of return of the risk-free asset. We will denote by $\mathcal{M}$ the set of equivalent\footnotemark  probability measures under which discounted stock prices are martingales. Furthermore, we will assume our market is free from arbitrage. Thus, we can assume that the set $\mathcal{M}$ is nonempty. For a complete market, $\mathcal{M}$ is a singleton, say $\mathcal{M}=\{\tP\}$ where $\tP$ is the unique probability measure under which discounted stock prices are martingales (see \cite{book} for more details about arbitrage, completeness, and equivalent martingale measures). We denote $\Psi^v$ as the set of self-financing portfolios given initial wealth $v$. 






\subsection{Weak Anticipation}

Now suppose we have some weak anticipation (weak information) regarding the prices of assets at our final time period. That is to say we know the distribution of $\vec{S_N}$. We will denote this distribution by $\nu$. Further, let $\mathcal{A}$ be the (finite) set of possible asset prices at time $N$. Note $|\mathcal{A}|\leq M^N$. Let $\mathcal{B}$ denote the set of all possible asset price processes over $N$ periods.

\begin{definition}:
The probability measure $\P^\nu$ defined by
\[\P^\nu(\o):=\sum\limits_{\vec{x}\in\mathcal{A}} \tP(\o|\vec{S_N}=\vec{x})\nu(\vec{S_N}=\vec{x})\] is called the minimal probability measure associated with the weak information $\nu$, where $\tP\in\mathcal{M}$ is an (remember $\mathcal{M}$ is a singleton in a complete market) equivalent martingale measure.
\end{definition}
In the sense of the following proposition, $\P^\nu$ is minimal in the set of probability measures $\Q$ equivalent to $\P$ such that $\Q(\vec{S_N}=\vec{x})=\nu(\vec{S_N}=\vec{x})$ for all $\vec{x}\in\mathcal{A}$. We denote this set by $\cE^{\nu}$. 

\begin{proposition}\label{Prop5gen}
Let $\phi$ be a convex function. Then
\[\min_{\Q\in\cE^\nu}\tE\left[\phi\left(\frac{d\Q}{d\tP}\right)\right]=\tE\left[\phi\left(\frac{d\P^\nu}{d\tP}\right)\right],\] where $\frac{d\Q}{d\tP}$ denotes the Radon-Nikodym derivative of $\Q$ with respect to $\tP$.
\end{proposition}

\footnotetext{In our finite discrete sample space, by equivalent we simply mean $\forall i \in \{1,2,..,M\}, \Q(\omega_i)>0.$}

\begin{proof}
Let $\vec{x}\in\mathcal{A}$ and $\Q\in\cE^\nu$ be given. Then,
\[ \tE\lrb{\frac{d\Q}{d\tP}|S_N=\vec{x}}=\frac{\nu(S_N=\vec{x})}{\tP(S_N=\vec{x})}.\]
Let $\phi$ be a convex function. Then, from conditional Jensen's inequality, \[\phi\left(\frac{\nu(S_N=\vec{x})}{\tP(S_N=\vec{x})}\right) = \phi\left(\tE\lrb{\frac{d\Q}{d\tP}|S_N=\vec{x}}\right) \leq \tE\lrb{\phi\left(\frac{d\Q}{d\tP}\right)|S_N=\vec{x}}.\]
Taking the expected value on both sides, we get \[\tE\lrb{\phi\left(\frac{\nu(S_N)}{\tP(S_N)}\right)}=\tE\lrb{\phi\left(\frac{d\P^\nu}{d\tP}\right)}\leq \tE\lrb{\phi\left(\frac{d\Q}{d\tP}\right)}\] and the result is proved.
\end{proof}

\subsection{Value of Weak Information}

Since an insider's anticipation has a different final time distribution than an uninformed investor, it is natural to find a way to characterize the value of this information. Since we focused on  maximizing our utility of wealth rather than the monetary value of wealth, we will define our value accordingly. 

\begin{definition}: 
 The \textbf{financial value of weak information} is the lowest expected utility that can be gained from anticipation. We write \[u(v,\nu)=\min\limits_{\Q\in\mathcal{E}^\nu} \max\limits_{\psi\in\Psi^v} \E^\Q[U(V_N)].\]
\end{definition}

Our main theorem is the following:

\begin{theorem}\label{weakval2}
    The financial value of weak information in a complete market is  \[u(v,\nu)=\max_{\psi\in\Psi^v} \E^\nu[U(V_N)]=\E^\nu \left[U\left(I\left(\frac{\l (v)}{(1+r)^N} \frac{d\tilde\P}{d\P^\nu}\right)\right)\right],\] where $\l (v)$ is determined by \[\tE\left[\frac{1}{(1+r)^N}I\left(\frac{\l (v)}{(1+r)^N}\frac{d\tilde\P}{d\P^\nu}\right)\right]=v,\] where $\tP \in \mathcal{M}$ is the unique probability measure under which the prices are martingales. Moreover, the optimal wealth at time $n$ $\hat{V}_n$ is given by
    \[\hat{V}_n=\frac{1}{(1+r)^{N-n}}\sum\limits_{b\in\mathcal{B}} I\left(\frac{\lambda(v)}{(1+r)^N}\frac{d\tP}{d\P^\nu}(b)\right)\tP(b |\vec{S}_n), \text{ for } n\in\{0,1,...,N\}.\] 
    The optimal amount to purchase of the $i^{th}$ linearly independent asset is
    \[\delta_n^i=\sum\limits_{j=1}^M (D_{n+1}^{-1})_{i,j} \hat{V}_{n+1}(\o_j), \text{ for } n\in\{0,1,...,N-1\},\]
    where 
    \[D_{n+1}=\begin{bmatrix} 
        S_{n+1}^1(\o_1) & S_{n+1}^2(\o_1) &.&.&.& S_{n+1}^M(\o_1) \\
        S_{n+1}^1(\o_2) & S_{n+1}^2(\o_2) &.&.&.& S_{n+1}^M(\o_2) \\
        . &.&&&& . \\
        . &.&&&& . \\
        . &.&&&& . \\
        S_{n+1}^1(\o_M) & S_{n+1}^2(\o_M) &.&.&.& S_{n+1}^M(\o_M) \\
    \end{bmatrix},\] is the matrix of $M$ linearly independent asset prices at time $n+1$, $(D_{n+1}^{-1})_{i,j}$ represents the element $(i,j)$ of the matrix $D_{n+1}^{-1}$, and $\hat{V}_{n+1}$ comes from the above equation.
\end{theorem}

\begin{proof}

We will proceed by rewriting $\max\limits_{\psi\in\Psi^v} \E^\Q[U(V_N)]$. In order to do this, we need the convex conjugate $\tilde U$, discussed in \cite{Utilde}.
We form the Lagrangian for solving $\max\limits_{\psi\in\Psi^v} \E^\Q[U(V_N)]$ by \[\mathcal{L}(\lambda)=\E^\Q[U(V_N)]+\lambda\left[v-\E^\Q\left[\frac{d\tP}{d\Q}\frac{V_N}{(1+r)^N}\right]\right].\] 
Now using $\tilde U$, substituting in for $V_N$ from the martingale method (see appendix), and doing algebra, we can rewrite our Lagrangian as


\[ \mathcal{L}(\lambda)= \lambda v + \tE\left[\frac{d\Q}{d\tP}\tilde U \left(\frac{\lambda}{(1+r)^N}\frac{d\tP}{d\Q}\right)\right].\]
Thus, we deduce 
\begin{align*}
u(v,\nu)&=\min\limits_{\Q\in\cE^\nu}\min\limits_{\lambda>0} \left[\lambda v + \tE\left[\frac{d\Q}{d\tP}\tilde U \left(\frac{\lambda}{(1+r)^N}\frac{d\tP}{d\Q}\right)\right]\right] \\ 
&=\min\limits_{\lambda>0} \left[\lambda v + \min\limits_{\Q\in\cE^\nu} \tE\left[\frac{d\Q}{d\tP}\tilde U \left(\frac{\lambda}{(1+r)^N}\frac{d\tP}{d\Q}\right)\right]\right].
\end{align*}
Since the convexity of $\tilde U$ implies the function mapping $z \mapsto z\tilde U\left(\frac{\lambda}{(1+r)^Nz}\right)$ is convex, 
we can use Proposition \ref{Prop5gen} to get 
\[u(v,\nu) = \min\limits_{\lambda>0} \left[\lambda v + \tE\left[\frac{d\P^\nu}{d\tP} \tilde{U} \left(\frac{\lambda}{(1+r)^N} \frac{d\tP}{d\P^\nu}\right)\right]\right].\]
Taking the derivative now with respect to $\lambda$ and setting it equal to 0, we find 
\[v = \tE \left[\frac{1}{(1+r)^N} I \left(\frac{\lambda^*(v)}{(1+r)^N} \frac{d\tP}{d\P^\nu}\right)\right]\] where $\lambda^*(v)$ is the minimizer. Now, 
\[u(v,\nu)=\lambda^*(v) v + \tE\left[\frac{d\P^\nu}{d\tP} \tilde{U} \left(\frac{\lambda^*(v)}{(1+r)^N} \frac{d\tP}{d\P^\nu}\right)\right] =
\E^\nu \left[U\left(I \left(\frac{\lambda^*(v)}{(1+r)^N} \frac{d\tP}{d\P^\nu}\right)\right)\right].\]
Thus, we have shown the first part of the theorem. Now note that discounted optimal wealth process $\{\frac{\hat{V}_n}{(1+r)^n}\}_{0\leq n\leq N}$ is a martingale under $\tP$. (see appendix) As a result, \[\hat{V}_n=\frac{1}{(1+r)^{N-n}}\tE[\hat{V}_N|\vec{S}_n]=\frac{1}{(1+r)^{N-n}}\sum\limits_{b\in\mathcal{B}} I\left(\frac{\lambda(v)}{(1+r)^N}\frac{d\tP}{d\P^\nu}(b)\right)\tP(b |\vec{S}_n)\]
for all $n\in\{0,1,...,N\}$.
Further, note that wealth is determined by your portfolio from the previous time period and the current prices. Thus, \[\hat{V}_{n+1}=D_{n+1}\vec{\d}_n,\] so we have \[\D_{n+1}^{-1}\hat{V}_{n+1}=\vec{\d}_n.\]
\end{proof}

\noindent \textit{Remark.} We know from \cite{book} that the matrix of all asset prices in the complete market has rank $M$. Therefore, we can choose $M$ linearly independent assets to invest in. Further, note that the optimal amount to purchase for each asset is only unique when $M=d$. \\

\begin{definition}:
    We define the \textbf{additional value of weak information} as the extra utility gained from investing with anticipation instead of just putting all of your wealth in the risk-free asset, which we define by \[F(v,\nu)=u(v,\nu)-U(v(1+r)^N).\]
\end{definition}
\begin{definition}:
    We also define the \textbf{ratio of added value to the total value} by \[\pi(v,\nu)=\frac{F(v)}{u(v,\nu)}=1-\frac{U(v(1+r)^N)}{u(v,\nu)}\]

\end{definition}


\section{Complete Markets: The Binomial Model} 

\subsection{Single-Period Binomial Model}
We first will focus on a one-period binomial model with two assets: one risk-free asset with payoff $1+r$, and one risky asset with payoffs $S_0(1+h)$ if the stock goes up, and $S_0(1-k)$ if the stock goes down. In order to have an arbitrage-free market, we require $h>r>-k$. Since there is only one risky asset, we will denote the amount of units owned of the risky asset at time $n$ by $\delta_n$. 

Figure \ref{fig:oneperiodlog} below shows a basic single-period binomial using the log utility. It represents the amount of stock you should buy initially, $\delta_0$.  From here there are only two outcomes for our final time; the stock price will either go up or down.       
 
\begin{center}
\begin{tikzpicture}[>=stealth,sloped]
    \matrix (tree) [%
      matrix of nodes,
      minimum size=.25cm,
      column sep=2cm,
      row sep=.25cm]
      { 
                    & {$\nu_0 = 50\%$} \\
    {$\delta_0 = 12.21095$}&                    \\
                    &{$\nu_1 = 50\%$}\\  
                    $n=0$             & $n=1$       \\
    };
   \draw[-](tree-2-1) -- (tree-3-2); 
   \draw[-](tree-2-1) -- (tree-1-2);
   \draw[->](tree-4-1) --(tree-4-2);
  \end{tikzpicture}
  \end{center}
  \begin{capt}\label{fig:oneperiodlog} An example of a single-period binomial model using the log utility where where the parameter values are $r = .032, h = .09, k = .019, v = 200.0,$ and $ s = 20.0 $
  \end{capt} 
  \noindent\textbf{Example 1} {Log Utility}\\
  When looking at the specific utility functions, in the case of log, we begin by maximizing $\E[U(V_N)]$ with respect to $\d$. We then are able to obtain our equation for the optimal number of shares with respect to wealth, $\hat{\d}$, in a one period model.
\[\hat{\delta}_0 = \frac{v(1+r)(\nu_0(h-r)+\nu_1(-k-r))}{-s(h-r)(-k-r)}.\]

\noindent\textbf{Example 2} {Power Utility}\\
As in log utility we would solve for our optimal number of shares with respect to wealth, $\hat{\d}_0$, in a one period model.
\[\hat{\d}_0 =\frac{((\nu_0(h-r))^{\frac{1}{\g-1}}-(\nu_1(-k-r))^{\frac{1}{\g-1}})(1+r)v}{(\nu_1(-k-r))^{\frac{1}{\g-1}}s(-k-r)-(\nu_0(h-r))^{\frac{1}{\g-1}}s(h-r)}.\]
\textbf{Example 3} {Exponential Utility}\\
Similarly to the previously examined utilities we will solve for the optimal number of shares with respect to wealth, $\hat{\d}$, in a one period model for the exponential utility.
\[\hat{\delta}_0 = \frac{\ln{(\nu_0 (h-r))}-\ln{(-\nu_1 (-k-r))}}{s(h+k)}.\]

\subsection{N-Period Binomial Model}

In binomial models, everything can be explicitly computed.   For instance, the following theorem gives the formula for the transition probabilities of the minimal probability $\mathbb{P}^\nu$. It is easy to establish by using the formula for conditional probabilities and straightforward combinatorial arguments. We note that $S_N$ is a Markov chain under the probability $\mathbb{P}^\nu$.

\begin{theorem}
Let $l\in\{1,...,N-1\}$ and $i\in\{0,...,N-l\}$. Then

\begin{align*}
\P^\nu(S_{N-l+1}=(1+h)S_{N-l}\mid  &S_{N-l}=(1+h^{N-l-i}(1-k)^is)\\
&= \frac{\sum\limits_{j=0}^{l-1}{\binom{l-1} {j}}(N-i-j)\ldots (N-i-(l-1))(i+1)(i+2)\ldots (i+j)\nu_{i+j}}{\sum\limits_{j=0}^{l}{\binom{l} {j}}(N-i-j)\ldots (N-i-(l-1))(i+1)(i+2)\ldots (i+j)\nu_{i+j}}\\ \intertext{ and }
\P^\nu(S_{N-l+1}=(1-k)S_{N-l}\mid  &S_{N-l}=(1+h^{N-l-i}(1-k)^is)\\
&= \frac{\sum\limits_{j=0}^{l-1}{\binom{l-1} {j}} \ldots (N-i-(l-1))(i+1)\ldots (i+j+1)\nu_{i+j+1}}{\sum\limits_{j=0}^{l}{\binom{l} {j}}(N-i-j)\ldots (N-i-(l-1))(i+1)(i+2)\ldots (i+j)\nu_{i+j}}.\\
\end{align*}
\end{theorem}

    \begin{capt}
    $\P^\nu$ for a 3-period binomial model
\end{capt}
\begin{center}
  \begin{tikzpicture}[>=stealth,sloped]
    \matrix (tree) [%
      matrix of nodes,
      minimum size=.01cm,
      column sep=2cm,
      row sep=.5cm]
      { $n=0$             & $n=1$                & $n=2$    &$n=3$ \\
                    &                   &       &$i=0$ \\
                    &                   &$\bullet$  & \\
                    & $\bullet$  &       &$i=1$\\
    $\bullet$&                   &$\bullet$  & \\
                    &$\bullet$   &       &$i=2$\\
                    &                   &$\bullet$  & \\
                    &                   &       &$i=3$\\
    };
    \draw[->] (tree-5-1) -- (tree-4-2) node [midway,above] {$\frac{3\nu_0+2\nu_1+\nu_2}{3\nu_0+3\nu_1+3\nu_2+3\nu_3}$};
    \draw[->] (tree-5-1) -- (tree-6-2) node [midway,above] {$\frac{\nu_1+2\nu_2+3\nu_3}{3\nu_0+3\nu_1+3\nu_2+3\nu_3}$};
    \draw[->] (tree-4-2) -- (tree-3-3) node [midway,above] {$\frac{3\nu_0+\nu_1}{3\nu_0+2\nu_1+\nu_2}$};
    \draw[->] (tree-4-2) -- (tree-5-3) node [midway,above] {$\frac{\nu_1+\nu_2}{3\nu_0+2\nu_1+\nu_2}$};
    \draw[->] (tree-6-2) -- (tree-5-3) node [midway,above] {$\frac{\nu_1+\nu_2}{\nu_1+2\nu_2+3\nu_3}$};
    \draw[->] (tree-6-2) -- (tree-7-3) node [midway,above] {$\frac{\nu_2+3\nu_3}{\nu_1+2\nu_2+3\nu_3}$};
    \draw[->] (tree-3-3) -- (tree-2-4) node [midway,above] {$\frac{3\nu_0}{3\nu_0+\nu_1}$};
    \draw[->] (tree-3-3) -- (tree-4-4) node [midway,above] {$\frac{\nu_1}{3\nu_0+\nu_1}$};
    \draw[->] (tree-5-3) -- (tree-4-4) node [midway,above] {$\frac{\nu_1}{\nu_1+\nu_2}$};
    \draw[->] (tree-5-3) -- (tree-6-4) node [midway,above] {$\frac{\nu_2}{\nu_1+\nu_2}$};
    \draw[->] (tree-7-3) -- (tree-6-4) node [midway,above] {$\frac{\nu_2}{\nu_2+3\nu_3}$};
    \draw[->] (tree-7-3) -- (tree-8-4) node [midway,above] {$\frac{3\nu_3}{\nu_2+3\nu_3}$};
  \end{tikzpicture}    
\end{center}

\begin{minipage} {\textwidth}
\begin{capt}
 $\P^\nu$ for a 3-period binomial model for a specific choice of $\nu$
\end{capt}
\begin{center}
  \begin{tikzpicture}[>=stealth,sloped,font=\scriptsize]
    \matrix (tree) [%
      matrix of nodes,
      minimum size=0.1cm,
      column sep=2.4cm,
      row sep=0.2cm
    ]
    {
    n=0             &n=1         &n=2        &n=3 \\
    \hspace{0.1cm}  &            &           &{$\nu_0=1/4$} \\
                    &            &{$\cdot$}  & \\
                    & {$\cdot$}  &           &{$\nu_1=1/2$}\\
    {$\cdot$}&                   &{$\cdot$}  & \\
                    &{$\cdot$}   &           &{$\nu_2=1/8$}\\
                    &            &{$\cdot$}  & \\
                    &            &           &{$\nu_3=1/8$}\\
    };
    \draw[->] (tree-5-1) -- (tree-4-2) node [midway,above] {$15/24$};
    \draw[->] (tree-5-1) -- (tree-6-2) node [midway,above] {$9/24$};
    \draw[->] (tree-4-2) -- (tree-3-3) node [midway,above] {$2/3$};
    \draw[->] (tree-4-2) -- (tree-5-3) node [midway,above] {$1/3$};
    \draw[->] (tree-6-2) -- (tree-5-3) node [midway,above] {$5/9$};
    \draw[->] (tree-6-2) -- (tree-7-3) node [midway,above] {$4/9$};
    \draw[->] (tree-3-3) -- (tree-2-4) node [midway,above] {$3/5$};
    \draw[->] (tree-3-3) -- (tree-4-4) node [midway,above] {$2/5$};
    \draw[->] (tree-5-3) -- (tree-4-4) node [midway,above] {$4/5$};
    \draw[->] (tree-5-3) -- (tree-6-4) node [midway,above] {$1/5$};
    \draw[->] (tree-7-3) -- (tree-6-4) node [midway,above] {$1/4$};
    \draw[->] (tree-7-3) -- (tree-8-4) node [midway,above] {$3/4$};
  \end{tikzpicture}
\end{center}
\end{minipage}

\newpage
\noindent\textbf{Example 1} (Log Utility)\\ 
\begin{center}
\begin{tikzpicture}[>=stealth,sloped,font=\scriptsize]
    \matrix (tree) [
      matrix of nodes,
      minimum size=.1cm,
      column sep=.1cm,
      row sep=.1cm]
      { 
                    &                   &       &{$S_3 = 25.90058$} \\
                    &                   &{$S_2 = 23.762$}  & \\
                    & {$S_1 = 21$}  &      &{$S_3 = 23.31052$}\\
    {$S_0 = 20$}&                   &{$S_2 = 21.3858$}  & \\
                    &{$S_1 = 19.62$}   &   &{$S_3 = 20.97947$}\\
                    &                   &{$S_3 = 19.24722$}  & \\
                    &                   &       &{$S_3=18.88152$}\\
                    $n=0$             & $n=1$                & $n=2$    &$n=3$ \\
    };
    \draw[-] (tree-4-1) -- (tree-3-2);
    \draw[-] (tree-4-1) -- (tree-5-2);
    \draw[-] (tree-3-2) -- (tree-2-3);
    \draw[-] (tree-3-2) -- (tree-4-3);
    \draw[-] (tree-5-2) -- (tree-4-3);
    \draw[-] (tree-5-2) -- (tree-6-3);
    \draw[-] (tree-2-3) -- (tree-1-4);
    \draw[-] (tree-2-3) -- (tree-3-4);
    \draw[-] (tree-4-3) -- (tree-3-4);
    \draw[-] (tree-4-3) -- (tree-5-4);
    \draw[-] (tree-6-3) -- (tree-5-4);
    \draw[-] (tree-6-3) -- (tree-7-4);
    \draw[->] (tree-8-1) -- (tree-8-2);
    \draw[->] (tree-8-2) -- (tree-8-3);
    \draw[->] (tree-8-3) -- (tree-8-4);
  \end{tikzpicture}
  \end{center}
  \begin{capt} A 3-period binomial tree showing the values of $S_n$ where the parameters are $r = .032, h = .09, k = .019, v = 200.0,$ and $ s = 20.0 $
  \end{capt} 
 \begin{multicols}{2}
\begin{center}
\begin{tikzpicture}[>=stealth,sloped,font=\scriptsize]
    \matrix (tree) [
      matrix of nodes,
      minimum size=.1cm,
      column sep=.1cm,
      row sep=.2cm]
      { 
                    &                   &       &{$\nu_0 = 25\%$} \\
                    &                   &{$\delta_2 = 146.4281$}  & \\
                    & {$\delta_1 = 76.48093$}  &      &{$\nu_1=25\%$}\\
    {$\delta_0 = 12.21095$}&                   &{$\delta_2 = 8.141736$}  & \\
                    &{$\delta_1 = -50.58155$}   &   &{$\nu_2=25\%$}\\
                    &                   &{$\delta_2=-107.9549$}  & \\
                    &                   &       &{$\nu_3=25\%$}\\
                    $n=0$             & $n=1$                & $n=2$    &$n=3$ \\
    };
    \draw[-] (tree-4-1) -- (tree-3-2);
    \draw[-] (tree-4-1) -- (tree-5-2);
    \draw[-] (tree-3-2) -- (tree-2-3);
    \draw[-] (tree-3-2) -- (tree-4-3);
    \draw[-] (tree-5-2) -- (tree-4-3);
    \draw[-] (tree-5-2) -- (tree-6-3);
    \draw[-] (tree-2-3) -- (tree-1-4);
    \draw[-] (tree-2-3) -- (tree-3-4);
    \draw[-] (tree-4-3) -- (tree-3-4);
    \draw[-] (tree-4-3) -- (tree-5-4);
    \draw[-] (tree-6-3) -- (tree-5-4);
    \draw[-] (tree-6-3) -- (tree-7-4);
    \draw[->] (tree-8-1) -- (tree-8-2);
    \draw[->] (tree-8-2) -- (tree-8-3);
    \draw[->] (tree-8-3) -- (tree-8-4);
  \end{tikzpicture}\columnbreak
  \begin{tikzpicture}[>=stealth,sloped,font=\scriptsize]
    \matrix (tree) [%
      matrix of nodes,
     minimum size=.1cm,
      column sep=.1cm,
      row sep=.2cm]
      { 
                    &                   &       &{$\nu = 20\%$} \\
                    &                   &{$\delta_2 = 251.9051$}  & \\
                    & {$\delta_1 = 192.0971$}  &      &{$\nu=40\%$}\\
    {$\delta_0 = 96.13333$}&                   &{$\delta_2=112.1887$}  & \\
                    &{$\delta_1 = -32.0822$}   &   &{$\nu=30\%$}\\
                    &                   &{$\delta_2=-224.84$}  & \\
                    &                   &       &{$\nu=10\%$}\\
                    $n=0$             & $n=1$                & $n=2$    &$n=3$ \\
    };
    \draw[-] (tree-4-1) -- (tree-3-2);
    \draw[-] (tree-4-1) -- (tree-5-2);
    \draw[-] (tree-3-2) -- (tree-2-3);
    \draw[-] (tree-3-2) -- (tree-4-3);
    \draw[-] (tree-5-2) -- (tree-4-3);
    \draw[-] (tree-5-2) -- (tree-6-3);
    \draw[-] (tree-2-3) -- (tree-1-4);
    \draw[-] (tree-2-3) -- (tree-3-4);
    \draw[-] (tree-4-3) -- (tree-3-4);
    \draw[-] (tree-4-3) -- (tree-5-4);
    \draw[-] (tree-6-3) -- (tree-5-4);
    \draw[-] (tree-6-3) -- (tree-7-4);
    \draw[->] (tree-8-1) -- (tree-8-2);
    \draw[->] (tree-8-2) -- (tree-8-3);
    \draw[->] (tree-8-3) -- (tree-8-4);
  \end{tikzpicture}
  \end{center}

   \end{multicols}
   \begin{capt}\label{fig:twologtrees} 3-period binomial trees showing the values of $\delta$ for various anticipations of $\nu$ using the log utility where the parameters are $r = .032$, $h = .09$, $k = .019$, $v = 200.0,$ and $ s = 20.0 $
  \end{capt}  
  
Figure \ref{fig:twologtrees} shows an example of two different 3-period binomial trees with set values. The first tree shows the values of $\delta$ at time $n$ when the anticipation has a uniform distribution. The second tree, however, shows an optimistic anticipation example. One can see how the amount of stocks one should invest changes depending on the distribution of the anticipation. For example, one would want to buy more stocks in an optimistic model because there is a better chance of the stock increasing in price as time goes on than in the model where all of the probabilities are the same. We allow for short-selling, so our $\delta$ can take negative values.

Further looking into the logarithmic utility function we can use Theorem 2.2 to find the financial value of weak information for the log utility function:
 We first solve for $\lambda$.
 \begin{align*}
  v &= \tE\left[\frac{1}{(1+r)^{N}} \cdot I\left(\frac{\lambda}{(1+r)^{N}} \frac{d\tilde{\mathbb{P}}}{d{\mathbb{P}}^{\nu}} \right)\right]\\
 v &= \tE\left[\frac{1}{(1+r)^{N}} \cdot\frac{(1+r)^{N}}{\lambda} \cdot \frac{d{\mathbb{P}}^{\nu}}{d\tilde{\mathbb{P}}}\right]\\
 \lambda &= \frac{1}{v}.\\
\end{align*}
Substituting for $\lambda$ into our value of weak information equation we thus have,
\begin{align*}
    u(v,\nu)&=\E^\nu \left[U \left(I\left( \frac{\lambda}{(1+r)^{N}} \cdot \frac{d\tilde{\mathbb{P}}}{d{\mathbb{P}}^{\nu}} \right) \right)\right]\\
    &= \E^\nu \left[\ln \left(\frac{(1+r)^{N}}{{\frac{1}{v}}} \cdot \frac{d\P^\nu}{d\tP} \right) \right]\\
    &= \ln \left( v(1+r)^{N}\right) + \E^{\nu} \left[ \ln\pfrac{d\P^\nu}{d\tP} \right]
\end{align*}
The additional value for log utility is \[F(v,\nu)=\E^{\nu} \left[ \ln\left(\frac{d{\mathbb{P}}^{\nu}}{d\tilde{\mathbb{P}}}\right) \right],\]
and the proportion is \[\pi(v,\nu)=\frac{\E^{\nu} \left[ \ln\left(\frac{d{\mathbb{P}}^{\nu}}{d\tilde{\mathbb{P}}}\right) \right]}{\ln \left( v(1+r)^{N}\right) + \E^{\nu} \left[ \ln\left(\frac{d{\mathbb{P}}^{\nu}}{d\tilde{\mathbb{P}}}\right) \right]}.\]

\noindent \textit{Remark.} Since we are only working at time $N$, we can write 
\[F(v,\nu)= \E^{\nu} \left[\ln\pfrac{d\nu}{d\tP_{\vec{S}_N}} \right].\] Notice this is the relative entropy of $\nu$ with respect to $\tP_{\vec{S}_N}$. \\

Note that $F(v,\nu)$ is only a function of $\nu$, so for any fixed $\nu$, we have that $F(v,\nu)$ is constant. Furthermore, $\pi(v,\nu)$ is a decreasing function of $v$ for any fixed $\nu$. As a result, the wealthier you are, the less proportion of utility you are gaining as a result of anticipation. 
In an 5-period binomial model, with the four anticipations below, we can look at the above functions as functions of $v$.
\begin{itemize}
    \item Precise: $\{0.01,0.01,0.01,0.95,0.01,0.01\}$
    \item Uniform Distribution: $\{1/6,1/6,1/6,1/6,1/6,1/6\}$
    \item Conservative: $\{0.1,0.2,0.2,0.2,0.2,0.1\}$
    \item Risk-Neutral: $\nu = \tP$
\end{itemize}

\begin{center}
    \includegraphics[scale = .27]{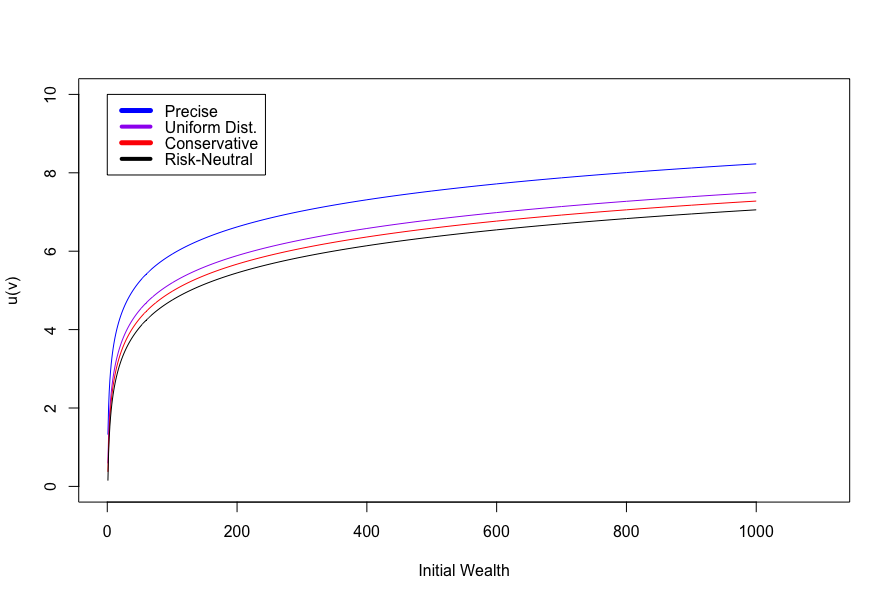}    
\end{center}
\begin{capt}
    Value of Weak Info. given $r=3\%$, $h=8\%$, $k=4\%$
\end{capt}
\begin{center}
    \includegraphics[scale = .27]{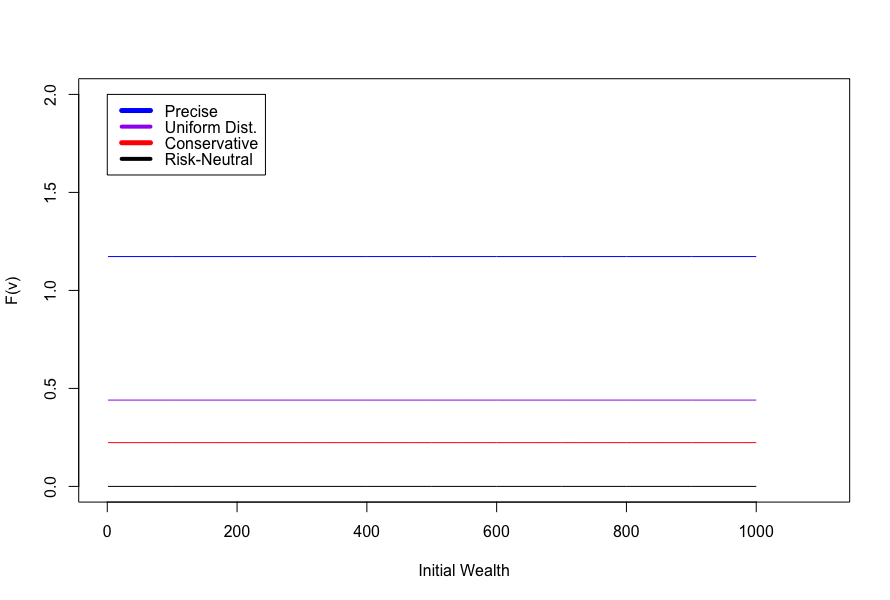}    
\end{center}
\begin{capt}
    Additional Value of Weak Info. given $r=3\%$, $h=8\%$, $k=4\%$
\end{capt}
\begin{center}
    \includegraphics[scale = .27]{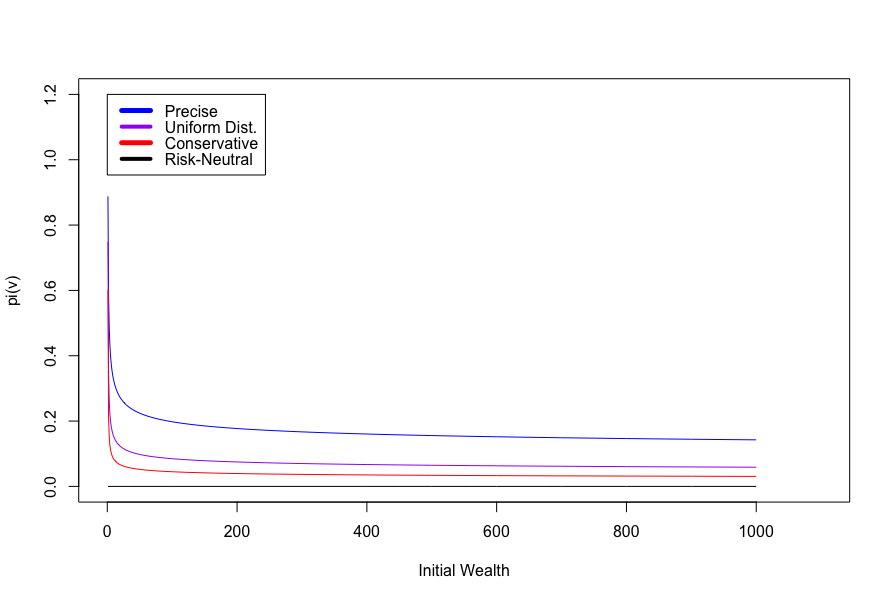}    
\end{center}
\begin{capt}
    Proportion of Value Added $r=3\%$, $h=8\%$, $k=4\%$
\end{capt}
\newpage
\textbf{Example 2} (Power Utility)\\ 
 
\begin{center}
\begin{tikzpicture}[>=stealth,sloped,font=\scriptsize]
    \matrix (tree) [
      matrix of nodes,
      minimum size=.1cm,
      column sep=.1cm,
      row sep=.1cm]
      { 
                    &                   &       &{$S_3 = 25.90058$} \\
                    &                   &{$S_2 = 23.762$}  & \\
                    & {$S_1 = 21$}  &      &{$S_3 = 23.31052$}\\
    {$S_0 = 20$}&                   &{$S_2 = 21.3858$}  & \\
                    &{$S_1 = 19.62$}   &   &{$S_3 = 20.97947$}\\
                    &                   &{$S_3 = 19.24722$}  & \\
                    &                   &       &{$S_3=18.88152$}\\
                    $n=0$             & $n=1$                & $n=2$    &$n=3$ \\
    };
    \draw[-] (tree-4-1) -- (tree-3-2);
    \draw[-] (tree-4-1) -- (tree-5-2);
    \draw[-] (tree-3-2) -- (tree-2-3);
    \draw[-] (tree-3-2) -- (tree-4-3);
    \draw[-] (tree-5-2) -- (tree-4-3);
    \draw[-] (tree-5-2) -- (tree-6-3);
    \draw[-] (tree-2-3) -- (tree-1-4);
    \draw[-] (tree-2-3) -- (tree-3-4);
    \draw[-] (tree-4-3) -- (tree-3-4);
    \draw[-] (tree-4-3) -- (tree-5-4);
    \draw[-] (tree-6-3) -- (tree-5-4);
    \draw[-] (tree-6-3) -- (tree-7-4);
    \draw[->] (tree-8-1) -- (tree-8-2);
    \draw[->] (tree-8-2) -- (tree-8-3);
    \draw[->] (tree-8-3) -- (tree-8-4);
  \end{tikzpicture}
  \end{center}
  \begin{capt} A 3-period binomial tree showing the values of $s$ where the parameters are $r = .032, h = .09, k = .019, v = 200.0,$ and $ s = 20.0 $
  \end{capt} 
 \begin{multicols}{2}
\begin{center}
\begin{tikzpicture}[>=stealth,sloped,font=\scriptsize]
    \matrix (tree) [
      matrix of nodes,
      minimum size=.1cm,
      column sep=.1cm,
      row sep=.1cm]
      { 
                    &                   &       &{$\nu_0 = 25\%$} \\
                    &                   &{$\delta_2 = 146.4281$}  & \\
                    & {$\delta_1 = 76.48093$}  &      &{$\nu_1=25\%$}\\
    {$\delta_0 = 12.21095$}&                   &{$\delta_2 = 8.141736$}  & \\
                    &{$\delta_1 = -50.58155$}   &   &{$\nu_2=25\%$}\\
                    &                   &{$\delta_2=-107.9549$}  & \\
                    &                   &       &{$\nu_3=25\%$}\\
                    $n=0$             & $n=1$                & $n=2$    &$n=3$ \\
    };
    \draw[-] (tree-4-1) -- (tree-3-2);
    \draw[-] (tree-4-1) -- (tree-5-2);
    \draw[-] (tree-3-2) -- (tree-2-3);
    \draw[-] (tree-3-2) -- (tree-4-3);
    \draw[-] (tree-5-2) -- (tree-4-3);
    \draw[-] (tree-5-2) -- (tree-6-3);
    \draw[-] (tree-2-3) -- (tree-1-4);
    \draw[-] (tree-2-3) -- (tree-3-4);
    \draw[-] (tree-4-3) -- (tree-3-4);
    \draw[-] (tree-4-3) -- (tree-5-4);
    \draw[-] (tree-6-3) -- (tree-5-4);
    \draw[-] (tree-6-3) -- (tree-7-4);
    \draw[->] (tree-8-1) -- (tree-8-2);
    \draw[->] (tree-8-2) -- (tree-8-3);
    \draw[->] (tree-8-3) -- (tree-8-4);
  \end{tikzpicture}
  \begin{capt}\label{fig:logvpower} Log Utility\end{capt}\columnbreak
  \begin{tikzpicture}[>=stealth,sloped,font=\scriptsize]
    \matrix (tree) [
      matrix of nodes,
      minimum size=.1cm,
      column sep=.1cm,
      row sep=.1cm]
      { 
                    &                   &       &{$\nu_0 = 25\%$} \\
                    &                   &{$\delta_2 = 1198.038$}  & \\
                    & {$\delta_1 = 445.6094$}  &      &{$\nu_1=25\%$}\\
    {$\delta_0 = 155.1425$}&                   &{$\delta_2 = 60..08356$}  & \\
                    &{$\delta_1 = 5.909925$}   &   &{$\nu_2=25\%$}\\
                    &                   &{$\delta_2 = -22.47464$}  & \\
                    &                   &       &{$\nu_3=25\%$}\\
                    $n=0$             & $n=1$                & $n=2$    &$n=3$ \\
    };
    \draw[-] (tree-4-1) -- (tree-3-2);
    \draw[-] (tree-4-1) -- (tree-5-2);
    \draw[-] (tree-3-2) -- (tree-2-3);
    \draw[-] (tree-3-2) -- (tree-4-3);
    \draw[-] (tree-5-2) -- (tree-4-3);
    \draw[-] (tree-5-2) -- (tree-6-3);
    \draw[-] (tree-2-3) -- (tree-1-4);
    \draw[-] (tree-2-3) -- (tree-3-4);
    \draw[-] (tree-4-3) -- (tree-3-4);
    \draw[-] (tree-4-3) -- (tree-5-4);
    \draw[-] (tree-6-3) -- (tree-5-4);
    \draw[-] (tree-6-3) -- (tree-7-4);
    \draw[->] (tree-8-1) -- (tree-8-2);
    \draw[->] (tree-8-2) -- (tree-8-3);
    \draw[->] (tree-8-3) -- (tree-8-4);
  \end{tikzpicture}
  \begin{capt}\label{fig:logvspower} Power Utility\end{capt}\end{center}
  \end{multicols}\begin{center}
  Two different 3-period binomial trees showing the values of $\delta$ for equal anticipations of $\nu$ using different utility where the constants are the same as Figure \ref{fig:twologtrees}. In the power utility model the value of $\gamma = .5$
  \end{center}
  
Figures \ref{fig:logvpower} and \ref{fig:logvspower} show the difference between the log and power utilities. As the log utility is a more relatively risk averse utility function (for $\gamma=0.5$), the absolute value of $\delta$ tends to be smaller when compared to the power utility function.  

Just like the log utility, we can also find the financial value of weak information for the power utility.  We now will solve for the value of $\lambda$.
  \[  \tE\left[\frac{1}{(1+r)^{N}} \cdot \left( \frac{\lambda}{(1+r)^{N}} \cdot \frac{d\tilde{\P}}{d\P^{\nu}}\right)^{\frac{1}{\gamma - 1}}\right]=v\\ \]
  \[  \lambda = \left(\frac{v(1+r)^{\frac{N\gamma}{\gamma-1}}}{\tilde{\E}\left[\pfrac{d\tilde{\P}}{d\P^{\nu}}^{\frac{1}{\gamma-1}} \right]} \right)^{\gamma-1}.\]
Substituting in for $\l$ we get,
\begin{align*}
    u(v,\nu)&=\E^\nu \left[U \left(I\left( \frac{\lambda}{(1+r)^{N}} \cdot \frac{d\tilde{\mathbb{P}}}{d{\mathbb{P}}^{\nu}} \right) \right)\right]\\
    &=\E^\nu\lrb{\frac{1}{\g}\lrp{\left(\frac{v(1+r)^{\frac{N\gamma}{\gamma-1}}}{\tilde{\E}\left[\pfrac{d\tilde{\P}}{d\P^{\nu}}^{\frac{1}{\gamma-1}} \right]} \right)^{\gamma-1} \cdot \frac{1}{(1+r)^N}\cdot\frac{d\tP}{d\P^\nu}}^{\frac{\g}{\g-1}}}\\
    &=\frac{v^{\gamma}(1+r)^{N\gamma}}{\gamma \left(\tilde{\E}\left[\left(\frac{d\tilde{\P}}{d\P^{\nu}}\right)^{\frac{1}{\gamma-1}}\right] \right)^{\gamma-1}} \cdot \E^{\nu}\left[\left(\frac{d\tilde{\P}}{d\P^{\nu}} \right)^{\frac{\gamma}{\gamma-1}}\right].
 \end{align*}
The additional value for power utility is \[F(v,\nu)= \frac{v^{\gamma}(1+r)^{N\gamma}}{\gamma \left(\tilde{\E}\left[\left(\frac{d\tilde{\P}}{d\P^{\nu}}\right)^{\frac{1}{\gamma-1}}\right] \right)^{\gamma-1}} \cdot \E^{\nu}\left[\left(\frac{d\tilde{\P}}{d\P^{\nu}} \right)^{\frac{\gamma}{\gamma-1}}\right]- \frac{v^\gamma(1+r)^{N\gamma}}{\gamma},\] 
and the proportion is \[\pi(v,\nu)=1-\frac{1}{\E^{\nu}\left[\left(\frac{d\tilde{\P}}{d\P^{\nu}} \right)^{\frac{\gamma}{\gamma-1}}\right] \cdot \left(\tilde{\E}\left[\left(\frac{d\tilde{\P}}{d\P^{\nu}}\right)^{\frac{1}{\gamma-1}}\right] \right)^{1-\gamma} }.\]


For power utility, we have the opposite relationship for a fixed $\nu$ with the proportion remaining constant  and the added value being an increasing function of initial wealth. 

\begin{center}
    \includegraphics[scale = .27]{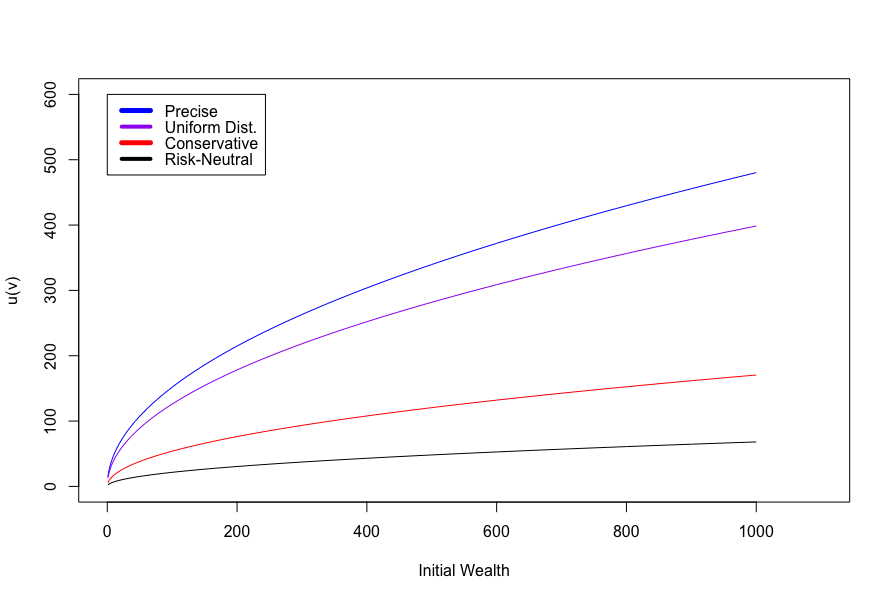}    
\end{center}
\begin{capt}
    Value of Weak Info. given $r=3\%$, $h=8\%$, $k=4\%$
\end{capt}
\begin{center}
    \includegraphics[scale = .27]{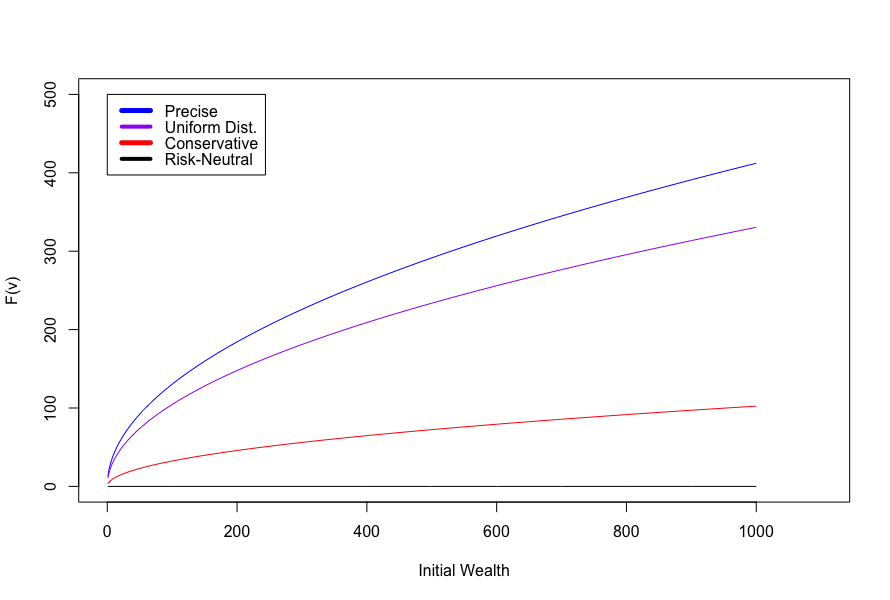}    
\end{center}
\begin{capt}
    Additional Value of Weak Info. given $r=3\%$, $h=8\%$, $k=4\%$
\end{capt}
\begin{center}
    \includegraphics[scale = .27]{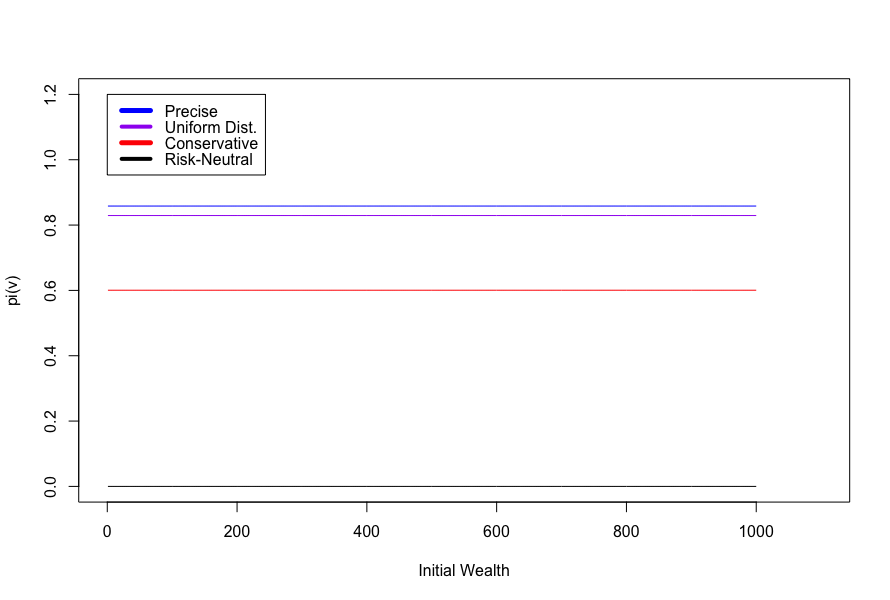}    
\end{center}
\begin{capt}
    Proportion of Value Added $r=3\%$, $h=8\%$, $k=4\%$
\end{capt}
\textbf{Example 3} (Exponential Utility)\\

\noindent We can also find the financial value of weak information for exponential utility.
\[\E^\nu\lrb{-e^{-a\hV_N}}=e^{-v\a(1+r)^N-\sum\limits_{i=0}^N \binom{N}{i} \tp^{N-i}\tq^i\ln \lrp{\binom{N}{i} \frac{\tp^{n-i}\tq^i}{\nu_i}}}.\]
We begin by solving for $\l$.
\[ \tE\lrb{\frac{1}{(1+r)^N}\cdot I\lrp{\frac{\l}{(1+r)^N}\frac{d\tP}{d\P^\nu}}}=v\\ \]
    We use this equation and then plug in for $I$.
    \[\tE\lrb{\frac{1}{(1+r)^N}\cdot\frac{-1}{\a}\cdot\ln\lrp{\frac{\l}{\a(1+r)^N}\frac{d\tP}{d\P^\nu}}}=v \]
We then solve for $\l$ to be:
   \[\l=\a(1+r)^Ne^{-v\a(1+r)^N-\E^\nu[\frac{d\tP}{d\P^\nu}\ln \pfrac{d\tP}{d\P^\nu}]}.\]
Finally we can use our $I$ and our $\l$ to plug in to our equation for the financial value of weak information to solve for the value as it specifically relates to exponential utility.
\begin{align*}
u(v,\nu)&=\E^\nu \left[U \left(I\left( \frac{\lambda}{(1+r)^{N}} \cdot \frac{d\tilde{\mathbb{P}}}{d{\mathbb{P}}^{\nu}} \right) \right)\right]\\
   &=\E^\nu\lrb{-e^{-a\frac{-1}{\a}\ln\lrp{\frac{\l }{\a(1+r)^N}\frac{d\tP}{d\P^\nu}}}}\\
    &=e^{-v\a(1+r)^N-\sum\limits_{i=0}^N \binom{N}{i} \tp^{N-i}\tq^i\ln \lrp{\binom{N}{i} \frac{\tp^{n-i}\tq^i}{\nu_i}}}.
\end{align*}




\section{Incomplete Markets: The Trinomial Model }

  \begin{center}\begin{tikzpicture}[>=stealth,sloped,font=\scriptsize]
    \matrix (tree) [%
      matrix of nodes,
      minimum size=.25cm,
      column sep=3cm,
      row sep=.35cm,
    ]
    {

         &       & $a^2s$\\
         &       &\\
         &$as$  &$abs$    \\
         &       &$acs$\\
    $s$  &$bs$  &$b^2s$ \ \footnote{}\\
         &       &$acs$\\
         &$cs$  &$bcs$     \\
         &       &\\
         &       &$c^2s$ \\
        $n=0$ &$n=1$  &$n=2$ \\
    };
    \draw[-] (tree-5-1) -- (tree-3-2) ;
    \draw[-] (tree-5-1) -- (tree-5-2) ;
    \draw[-] (tree-5-1) -- (tree-7-2) ;
    \draw[-] (tree-3-2) -- (tree-1-3) ;
    \draw[-] (tree-3-2) -- (tree-3-3) ;
    \draw[-] (tree-3-2) -- (tree-4-3) ;
    \draw[-] (tree-5-2) -- (tree-3-3) ;
    \draw[-] (tree-5-2) -- (tree-5-3) ;
    \draw[-] (tree-5-2) -- (tree-7-3) ;
    \draw[-] (tree-7-2) -- (tree-6-3) ;
    \draw[-] (tree-7-2) -- (tree-7-3) ;
    \draw[-] (tree-7-2) -- (tree-9-3);
    \draw[->] (tree-10-1) -- (tree-10-2);
    \draw[->] (tree-10-2) -- (tree-10-3);
  \end{tikzpicture}  \end{center}
  
\footnotetext{The value of $b^2s$ is not necessarily above or below $acs$, and the placement depends on the values of $a,b$, and $c.$}

For the general case of the trinomial model, we let our model be $N$ periods with sample space $\Omega=\{\o_1,\o_2,\o_3\}$ denoting taking the upper, middle, and bottom paths respectively. Let $i$ be the number of times the upper path with payoff $a$ is taken, and $j$ be the number of times the middle path with payoff $b$ is taken. Then $i=0,\ldots,N$ and $j=0,\ldots,N-i$. For any particular endpoint such that our stock went up $i$ times to get there, there are $N-i+1$ possible options for $j$. Since each unique endpoint is uniquely determined in terms of $i$ and $j$, and $i=0,\dots,N$, then our total number of possible endpoints is equal to \[\sum\limits_{i=1}^{N+1} i.\]

It is important to note the differences between the trinomial model, or for that matter any incomplete model, compared to the binomial model. To start, $\mathcal{M}$ is no longer a singleton. More specifically, \[\mathcal{M}=\{\tP: \tP(\o_1)(a-c)+\tP(\o_2)(b-c)=1+r-c, \tP(\o_3) = 1 - \tP(\o_1)-\tP(\o_2)\}, \forall i \in \{1,2,3\} 0<\tP(\o_i)<1.\] Knowing that we do not have a unique martingale measure $\tP$, we can not have complete replication, so we can not use the martingale method in the same way we did for the complete  market. The following is a method that accounts for our lack of a unique martingale measure $\tP$.

\subsection{Optimizing Utility}\label{5.1}
Let \[\mathcal{M}_n=\{\tP_n:\tP_n \text{ is an equivalent martingale measure for the $n^{th}$ period}\}.\] Note, since we are looking at martingale measures, $\tP_n$, the period and the paths taken in previous periods do not effect $\tP_n$. Let $\tP_{n,0}$ be the extremal measure such that for $b<1+r$ \[\tP_{n,0}(\o) = \begin{cases} 
\frac{(1+r)-b}{a-b}, &\o=\o_1 \\
\frac{a-(1+r)}{a-b}, &\o=\o_2 \\
0, &\o=\o_3
\end{cases}.\] For $b\geq1+r$ \[\tP_{n,0}(\o) = \begin{cases}
0, &\o=\o_1\\
\frac{\r-c}{b-c}, &\o=\o_2\\
\frac{b-\r}{b-c}, &\o=\o_3
\end{cases}. \] Next let $\tP_{n,1}(\o)$ be the extremal measure such that \[\tP_{n,1}(\o)=\begin{cases}
\frac{\r-c}{a-c}, &\o=\o_1\\
0, &\o=\o_2\\
\frac{a-\r}{a-c}, &\o=\o_3
\end{cases}.\] See \cite{IlamV}. Note that $\tP_{n,0},\tP_{n,1}\notin \mathcal{M}_n$. However, $\forall \tP_n\in\mathcal{M}_n$, $\tP_n$ is a convex combination \[t_n\tP_{n,0}+(1-t_n)\tP_{n,1}\] of $\tP_{n,0}\tnd\tP_{n,1}$ with $t_n\in(0,1)$ for all $n$. Note, $t_n$ may depend on $n$ as well as what occurred in previous periods. Then let $b\in\mathcal{B}$. We continue with \[\tP(b)=\prod\li_{i=0}^{N-1}(t_i\tP_{i,0}+(1-t_i)\tP_{i,1}).\] It follows for $N=2$ that \[\tP(b)=t_1t_2\tP_{0,0}\tP_{1,0}+(1-t_1)t_2\tP_{0,1}\tP_{1,0}+t_1(1-t_2)\tP_{0,0}\tP_{1,1}+(1-t_1)(1-t_2)\tP_{0,1}\tP_{1,1}.\] We then define 
   \[ \tP^1:=\tP_{0,0}\tP_{1,0} \qquad\tP^2:=\tP_{0,1}\tP_{1,0}  \qquad \tP^3:=\tP_{0,0}\tP_{1,1} \qquad\tP^4:=\tP_{0,1}\tP_{1,1}.\]
We continue for the general $N$-period model with the obvious extension. Thus, $\tP(b)$ is a convex combination of $\tP^j$ for $j\in\{1,2,...,2^N$\}.\\
With respect to anticipation our Radon-Nikodym derivatives are defined as $ \frac{d\tP^j}{d\nu}(\o)$ for $j\in\{1,2,...,2^N\}.$ Notice, since $\tP$ is a convex combination of $\tP^j$ for $j\in\{1,2,...,2^N\}$, \[\tE\lrb{\frac{V_N}{(1+r)^N}}=v \qquad \forall \tP \iff \E^{\tP^j}\lrb{\frac{V_N}{(1+r)^N}}=v \qquad \forall j\in\{1,2,...,2^N\}.\]
We can then solve the Lagrangian (see \cite{IlamV}) for $\hat{V}_N$, we find
\[\hat{V}_N(b)= I\lrp{\sum_{j=1}^{2^N}\frac{\l_j}{(1+r)^N}\frac{d\tP^j}{d\nu}(b)},\]
where the $\l_j$'s satisfy \[v=\E^\nu\lrb{\frac{1}{(1+r)^N}\frac{d\tP^i}{d\nu}I\lrp{\sum_{j=1}^{2^N}\frac{\l_j}{(1+r)^N}\frac{d\tP^j}{d\nu}}} \text { for each } i\in\{1,2,...,2^N\}.\] By \cite{Pliska} we know that our $\l_j$'s are unique due to the concavity of $U(x)$.

\noindent\textbf{Example 1} (Log Utility)\\
When optimizing our trinomial model from $V_N$ we plug in for $I$ given our specific utility function.
\begin{align*}
    \hat{V}_N(b)&= \frac{(1+r)^N}{\sum\li_{j=1}^{2^N}\l_j\frac{\tP^j(b)}{\nu(b)}},
\end{align*}
where the $\l_j$'s satisfy
\[v=\E^\nu\lrb{\frac{\tP^i(b)}{\sum\li_{j=1}^{2^N}\l_j\tP^j(b)}} \text { for each } i\in\{1,2,...,2^N\}.\]
\textbf{Example 2} (Power Utility)\\
Similarly in power utility we use our function to plug in for $I$.
\begin{align*}
    \hat{V}_N(b)&=\lrp{\sum_{j=1}^{2^N} \frac{\l_j}{(1+r)^N}\frac{\tP^j(b)}{\nu(b)}}^{\frac{1}{\g-1}},
\end{align*}
where the $\l_j$'s satisfy 
\[v=\E^\nu\lrb{\frac{1}{(1+r)^N}\frac{d\tP^i}{d\nu}\lrp{\sum_{j=1}^{2^N}\frac{\l_j}{(1+r)^N}\frac{d\tP^j}{d\nu}}^{\frac{1}{\g-1}}} \text { for each } i\in\{1,2,...,2^N\}.\]
\textbf{Example 3} (Exponential Utility)\\
In Exponential utility our steps would be similar to the previous functions. We plug in for $I$.
\[\hat{V}_N(b)=-\frac{1}{\a}\ln\lrp{\frac{-1}{\r^N\a}\sum\li_{j=1}^{2^N}\l_j\frac{d\tP^j}{d\nu}(b)},\]
where the $\l_j$'s satisfy
\[v=\E^\nu\lrb{-\frac{1}{\r^N\a}\frac{d\tP^i}{d\nu}\ln\lrp{\frac{-1}{\r^N\a}\sum\li_{j=1}^{2^N}\l_j\frac{d\tP^j}{d\nu}}}\text { for each } i\in\{1,2,...,2^N\}.\]
\subsection{Finding an Optimal Portfolio}
Understanding that incomplete markets cannot be replicated, our approach to finding an optimal portfolio needs to change accordingly. Recall from \ref{5.1} \[\tE\lrb{\frac{V_N}{(1+r)^N}}=v \qquad \forall \tP \iff \E^{\tP^j}\lrb{\frac{V_N}{(1+r)^N}}=v \qquad \forall j\in\{1,2,...,2^N\}.\] Thus, given this constraint, if $\E^{\tP^j}\lrb{\frac{V_N}{(1+r)^N}}=v, \forall j\in\{1,2,...,2^N\},$ our $\hat{V}_n$ can be replicated by a self-financing portfolio. Given this fact, we can determine our optimal portfolio strategy, $\dseq {N-1}$:
\[\hd_n=\frac{\hV_{n+1}(\o_i)-\hV_{n+1}(\o_j)}{S_{n+1}(\o_i)-S_{n+1}(\o_j)}, \text{ for } i,j\in\{1,2,3\}, i\neq j, n \in \{0,1,...,N-1\}.\] Fixing $\tP\in\cM$, we can find \[\hV_n = \frac{1}{(1+r)^{N-n}} \cdot \tE\left[I\left(\sum\limits_{j=1}^{2^N}\frac{\l_j}{(1+r)^N}\frac{d\tP^j}{d\nu}\right)|S_n\right],\] to plug in to our expression solving for $\hd$. 
For further explanation see \cite{IlamV}.

\section{Appendix}

The following is with respect to the general discrete case in a complete market. As in Section \ref{2}, we denote $\Psi^v$ as the set of self-financing portfolios given initial wealth $v$. 
\begin{theorem}\label{Vmg}
    The discounted wealth process is a martingale under the martingale measure $\Q$.
\end{theorem}
\begin{proof} 
    See \cite{IlamV}.
\end{proof}
\begin{theorem}
    Maximizing $\E[U(V_N)]$ over the set of self-financing portfolios $\Psi^v$ is equivalent to maximizing $\E[U(V_N)]$ subject to $\tE[U(V_N)]=v$, with $\tP$ being the unique equivalent martingale measure.
\end{theorem} 

\begin{proof}
    See \cite[Lemma 4.9]{maxthm}.
\end{proof}

\begin{theorem}
    \[V_N=I\left(\frac{\lambda}{(1+r)^N}\frac{d\tP}{d\Q}\right)\] More specifically, optimal terminal wealth $\hat V_N$ is attained when $\lambda$ satisfies \[v = \tE\left[\frac{1}{(1+r)^N} I\left(\frac{\lambda}{(1+r)^N}\frac{d\tP}{d\Q}\right)\right].\]
\end{theorem}
\begin{proof}
    See \cite{IlamV} p16.
\end{proof}

\bibliography{SummerWorksCited}
\bibliographystyle{plain}
\nocite{*}

\end{document}